\documentclass[11pt]{article}
\usepackage[T1]{fontenc}
\usepackage{lmodern}
\usepackage[margin=1in]{geometry}
\usepackage{amsmath,amssymb,amsthm,mathtools}
\usepackage{microtype}
\usepackage{tikz}
\usetikzlibrary{arrows.meta}
\definecolor{resultblue}{RGB}{35,83,144}
\usepackage[colorlinks=true,linkcolor=blue,citecolor=blue,urlcolor=blue]{hyperref}
\hypersetup{pdftitle={EFI iff 1PRS},pdfauthor={Rishabh Batra}}
\usepackage[nameinlink,noabbrev]{cleveref}
\numberwithin{equation}{section}
\newtheorem{theorem}{Theorem}[section]
\newtheorem{lemma}[theorem]{Lemma}

\newtheorem{corollary}[theorem]{Corollary}
\newtheorem{fact}[theorem]{Fact}
\theoremstyle{definition}
\newtheorem{definition}[theorem]{Definition}
\theoremstyle{remark}

\newcommand{\EFI}{\ensuremath{\mathsf{EFI}}}
\newcommand{\PRS}{\ensuremath{\mathsf{1PRS}}}
\newcommand{\QPT}{\mathsf{QPT}}
\newcommand{\Tr}{\operatorname{Tr}}
\newcommand{\negl}{\operatorname{negl}}
\newcommand{\MM}[1]{\mathsf{MM}_{#1}}
\newcommand{\adv}{\operatorname{Adv}}
\newcommand{\capprox}{\mathrel{\approx_c}}
\newcommand{\sapprox}{\mathrel{\approx_s}}
\newcommand{\ket}[1]{\lvert #1\rangle}
\newcommand{\bra}[1]{\langle #1\rvert}
\newcommand{\poly}{\operatorname{poly}}
\title{$\EFI\Leftrightarrow\PRS$}
\author{Rishabh Batra\\EPFL}
\date{}
\begin{document}
\maketitle

\begin{abstract}
\noindent
A series of recent works has established EFI pairs as the ``minimal'' assumption
for quantum cryptography. Recently, Cavalar et al.~\cite{ccc} showed that EFI pairs are
equivalent to non-uniform expanding single-copy pseudorandom state generators
($\PRS$), and obtained a non-uniform meta-complexity characterization of EFI
pairs. We remove the non-uniformity in their construction and prove
$\EFI\Longleftrightarrow\PRS$. As a corollary, we obtain the first uniform
meta-complexity characterization of EFI pairs. Together with the equivalence
of statistically verifiable one-way state generators and commitments proved
by Batra and Jain~\cite{bj}, this gives a quantum analogue of the classical equivalence between
one-way functions, commitments, and pseudorandom generators.
\end{abstract}

\section{Introduction}

One-way functions are a central foundation of classical computational
cryptography: they are easy to evaluate but hard to invert, and their existence
is equivalent to that of pseudorandom generators and computationally secure
commitments~\cite{hill,naor,il}. Another equivalent primitive is an
\emph{EFID pair}~\cite{gol90}: a pair of efficiently samplable classical
distributions that are statistically far but computationally
indistinguishable. Thus one-way functions, EFID pairs, and pseudorandom
generators are equivalent existence assumptions~\cite{gol01}.
These equivalences connect inversion, pseudorandomness, and hiding, and
organize a large part of classical cryptography around a single assumption.

In the quantum world, the picture is less straightforward. Pseudorandom
state generators (PRSGs)~\cite{jls} generalize pseudorandom generators to
quantum outputs. Morimae and Yamakawa~\cite{my} showed that they suffice
for quantum commitments. Further cryptographic applications of
pseudorandom quantum states were obtained by Ananth, Qian, and
Yuen~\cite{aqy}.
It is known that PRSGs can exist even if $\mathsf{BQP}=\mathsf{QMA}$
(relative to a quantum oracle)~\cite{kre}, and that single-copy PRSGs can
exist even if $\mathsf P=\mathsf{NP}$ (relative to a classical
oracle)~\cite{kqst}. Both settings rule out quantum-secure classical
one-way functions relative to the corresponding oracles. These
separations motivate the possibility of doing quantum cryptography under
assumptions weaker than OWFs.

EFI pairs, introduced by Brakerski, Canetti, and Qian~\cite{bcq}, are the
quantum analogue of EFID pairs: they consist of efficiently generated
quantum states that are statistically far but computationally
indistinguishable. They are equivalent to quantum commitments and are
widely regarded as a candidate minimal assumption for quantum
cryptography (see~\cite{microcrypt}).

Morimae and Yamakawa~\cite{my} introduced one-way state generators
(OWSGs), quantum analogues of OWFs with quantum outputs. Batra and
Jain~\cite{bj} showed that EFI pairs are equivalent to a version of OWSGs
whose verifier may be inefficient, called statistically verifiable
one-way state generators ($\mathsf{svOWSG}$s).

In the single-copy setting, a short classical key specifies a pure state
whose average is computationally indistinguishable from the maximally
mixed state. We use $\PRS$ to mean an \emph{expanding} single-copy
pseudorandom state generator: its output has more qubits than its key has
bits. Without expansion, the single-copy condition can be satisfied
trivially: mapping a uniformly random $n$-bit key $k$ to the basis state
$\ket k$ gives an average state exactly equal to $I/2^n$.
Cavalar et al.~\cite{ccc} proved
$\EFI\Longleftrightarrow\mathsf{nu\text{-}1PRS}$, where the generator
receives $O(\log\lambda)$ bits of classical advice specifying an
extraction length. Removing this advice was a major question left open.

\subsection{Our results}

Our main result resolves this question.

\begin{theorem}[Informal]
EFI pairs exist if and only if expanding single-copy pseudorandom state
generators exist.
\end{theorem}

Meta-complexity studies the difficulty of determining an object's own
complexity. Examples include estimating the length of a shortest program
that describes a string and finding the smallest circuit for a truth table.
These questions connect cryptography to learning, compression, and circuit
lower bounds. Classically, average-case hardness of suitable Kolmogorov
complexity problems gives a \textbf{uniform meta-complexity characterization of one-way
functions}~\cite{lp,irs}.  Consequences for learning illustrate why such characterizations
are useful beyond constructing cryptographic primitives~\cite{hn}.

For quantum states, Cavalar et al.~\cite{ccc} give a uniform
meta-complexity characterization of $\PRS$ using a robust version of
\emph{G\'acs' quantum algorithmic entropy}~\cite{gacs}.
This entropy measures the algorithmic description complexity of a quantum
state. Its lower version is
\begin{equation*}
 \underline{\mathsf H}(\ket\psi)
 =-\log\bra\psi\mu_n\ket\psi,
\end{equation*}
where $\mu_n$ is a universal semi-density matrix that assigns greater
weight to states with shorter descriptions. A pure state can therefore
have high algorithmic entropy despite having zero von Neumann entropy.
The robust version $\underline{\mathsf H}^{\delta}$ maximizes this
quantity over pure states within trace distance $\delta$.
The promise problem $\mathsf{GapH}^{\varepsilon}(u,u+\Delta)$ asks,
given a single copy, to distinguish
\begin{equation*}
 \underline{\mathsf H}^{1-\varepsilon}(\ket\psi)\leq u
 \qquad\text{from}\qquad
 \underline{\mathsf H}(\ket\psi)\geq u+\Delta.
\end{equation*}
For a keyed family sampled with a uniform key, each promise event must
have probability at least $1/2-\negl(\lambda)$. Hardness means that the
two state ensembles conditioned on these events are computationally
indistinguishable given one copy.

\begin{theorem}[Meta-complexity characterization of $\PRS$~\cite{ccc}]
\label{thm:prs-meta}
Expanding single-copy pseudorandom state generators exist if and only if
there are a keyed family of pure states generated by a uniform QPT
algorithm, polynomial-time computable integer thresholds $u(\lambda)$ and
$\Delta(\lambda)=\omega(\log\lambda)$, and a constant
$0<\varepsilon<1/100$ such that
\begin{equation*}
 \mathsf{GapH}^{\varepsilon}
 \bigl(u(\lambda),u(\lambda)+\Delta(\lambda)\bigr)
\end{equation*}
is hard on average given one copy of a state whose key is sampled
uniformly, with the balanced promise and hardness requirements of~\cite{ccc}.
\end{theorem}

Cavalar et al.~\cite{ccc} also prove that EFI pairs are equivalent to
non-uniform expanding single-copy pseudorandom state generators
($\mathsf{nu\text{-}1PRS}$), yielding a non-uniform meta-complexity
characterization of EFI pairs. Our equivalence
$\EFI\Longleftrightarrow\PRS$ removes this non-uniformity and gives the
\textbf{first uniform meta-complexity characterization of EFI pairs}.

\begin{corollary}[Uniform meta-complexity characterization of EFI]
\label{cor:uniform-meta}
EFI pairs exist if and only if there are a keyed family of pure states
generated by a uniform QPT algorithm and parameters satisfying the
$\mathsf{GapH}^{\varepsilon}$ hardness condition in
Theorem~\ref{thm:prs-meta}.
\end{corollary}

Goldin et al.~\cite{gmmy} distinguish three regions of quantum
cryptography: QuantuMania, CountCrypt, and NanoCrypt. The NanoCrypt
primitives they list are EFI pairs, quantum non-interactive bit
commitments, secure multiparty computation, and single-copy pseudorandom
state generators. Our new implication shows that all of these are
equivalent as existence assumptions. Together with the known
equivalences~\cite{bcq,bj}, this yields the picture in
Figure~\ref{fig:microcrypt-equivalences}.

\begin{figure}[htbp]
\centering
\begin{tikzpicture}[
 primitive/.style={draw=resultblue!65,fill=resultblue!5,rounded corners=3pt,
   minimum height=12mm,text width=30mm,align=center,font=\small},
 known/.style={{Latex[length=2.2mm,width=1.6mm]}-{Latex[length=2.2mm,width=1.6mm]},line width=0.8pt},
 directed/.style={-{Latex[length=2.2mm,width=1.6mm]},line width=0.8pt},
 labeltext/.style={font=\scriptsize,inner sep=1pt}
]
\node[primitive] (prs) at (0,0) {Expanding $\PRS$};
\node[primitive] (efi) at (5.3,0) {EFI pairs};
\node[primitive] (owsg) at (10.6,0) {$\mathsf{svOWSG}$};
\node[primitive] (ot) at (0,-2.5) {Quantum oblivious\\transfer};
\node[primitive] (com) at (5.3,-2.5) {Quantum non-interactive\\bit commitments};
\node[primitive] (mpc) at (10.6,-2.5) {Secure multiparty\\computation};
\draw[directed,line width=1pt,resultblue]
 ([yshift=3mm]efi.west) -- node[labeltext,above=9pt,text=resultblue]
 {\textbf{This work}} ([yshift=3mm]prs.east);
\draw[directed]
 ([yshift=-3mm]prs.east) -- node[labeltext,below=9pt]
 {\cite{my,ccc}} ([yshift=-3mm]efi.west);
\draw[known] (efi) -- node[labeltext,above=3pt] {\cite{bj}} (owsg);
\draw[known] (efi) -- node[labeltext,right=3pt] {\cite{bcq}} (com);
\draw[known] (com) -- node[labeltext,above=3pt] {\cite{bcq}} (ot);
\draw[known] (com) -- node[labeltext,above=3pt] {\cite{bcq}} (mpc);
\end{tikzpicture}
\caption{Equivalences at the bottom of Microcrypt.}
\label{fig:microcrypt-equivalences}
\end{figure}
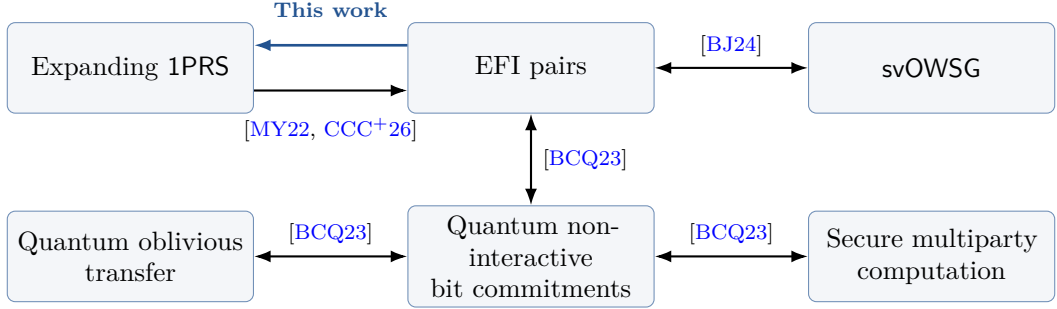

In particular,
\begin{equation*}
 \EFI\quad\Longleftrightarrow\quad
 \mathsf{svOWSG}\quad\Longleftrightarrow\quad
 \text{quantum commitments}\quad\Longleftrightarrow\quad\PRS.
\end{equation*}
This is a quantum analogue of the classical equivalence between one-way
functions, EFID pairs, commitments, and pseudorandom generators.

\subsection{Proof overview}

We first describe the construction of Cavalar et al.~\cite{ccc} and where
it needs advice, and then explain how to remove that advice. Let $\lambda$ be the security parameter. We start with an EFI pair
$(\rho_0,\rho_1)$ of trace distance $1-\negl(\lambda)$.

Write
\begin{equation*}
 \bar\rho=\frac{\rho_0+\rho_1}{2},\qquad
 h=\frac{S(\rho_0)+S(\rho_1)}2,\qquad
 \chi=S(\bar\rho)-h.
\end{equation*}
Statistical separation gives $\chi\geq1-\negl(\lambda)$.
The correlated-flag construction gives indistinguishable states
\begin{equation*}
 \alpha_{RF}=\frac12\sum_{b=0}^1\rho_b\otimes\ket b\bra b_F,
 \qquad \beta_{RF}=\bar\rho\otimes\frac{I_F}{2},
\end{equation*}
with $S(\alpha)=1+h$ and $S(\beta)=1+h+\chi$.
After obtaining this entropy gap, CCC+ use an asymptotic equipartition
argument and a strong quantum extractor implemented by Clifford
unitaries. Entropy-to-extraction arguments also appear in~\cite{bj}.
The gap supplies room for extraction, but choosing the extraction length
requires locating the unknown baseline entropy. Their $O(\log\lambda)$
bits of advice specify that length.

\paragraph{Purification of the EFI.}
Purify the EFI generator and keep its purification register. If
$\ket{\psi_b}_{RE}$ purifies $\rho_b$, define
\begin{equation*}
 G(b)=\ket b_F\ket{\psi_b}_{RE},\qquad
 \tau=\frac12\sum_{b=0}^1G(b)G(b)^\dagger.
\end{equation*}
The fixed-key outputs are pure and orthogonal, so $S(\tau)=1$ exactly.
We therefore know the entropy of this source without estimating the
entropy of either EFI state. To obtain expansion, we would like a
computationally indistinguishable source of higher entropy and then to
extract from many copies. Revealing the purification prevents us from
simply substituting $\beta$ for $\alpha$ globally. Instead, we obtain
the needed entropy surplus through local prefix witnesses.

\paragraph{Establishing entropy surplus.}
Order the source registers as $R,F,E$ and pad with zero qubits to a
power-of-two width $M\geq4$. Let $Q_f=F$, and let $\tau^{[j]}$ be the
marginal on the first $j$ qubits. The ordering puts the purification
after the flag, so $\tau^{[f]}=\alpha$. Define
\begin{equation*}
 \gamma_j=\tau^{[j]}\quad(j\neq f),\qquad \gamma_f=\beta.
\end{equation*}
For every $j$, the witness is computationally indistinguishable from
the actual prefix and has exactly the actual preceding marginal:
\begin{equation*}
 \gamma_j\capprox\tau^{[j]},\qquad
 \Tr_{Q_j}\gamma_j=\tau^{[j-1]}.
\end{equation*}
The entropy chain rule now gives, for
$h_j=H(Q_j\mid Q_{<j})_{\gamma_j}$,
\begin{equation*}
 \sum_{j=1}^{M}h_j
 =S(\tau)+S(\beta)-S(\alpha)
 =1+\chi\geq2-\negl(\lambda).
\end{equation*}
This is the only entropy sum we need. The witnesses need not be
marginals of a single global state, and their individual conditional
entropies may be negative. Their exact agreement with the actual past
will let us use them in consecutive hybrids.

\paragraph{Random-start idea.}
A lower bound on the sum of the $h_j$ is known, but its summands are not.
To avoid choosing
a separate extraction length for each source coordinate, use the
random-start idea of~\cite{hrv}. Concatenate $L=4M$
independently keyed source copies, choose
$J\in\{0,\ldots,M-1\}$ uniformly, and take the length-$N=M(L-1)$
window beginning at position $J+1$.

For its $i$th prefix, replace only the current source copy by the
corresponding local witness. Denote the actual and witness states by
$\omega_i$ and $\Omega_i$. They satisfy
\begin{equation*}
 \Omega_i\capprox\omega_i,\qquad
 \Tr_{A_i}\Omega_i=\omega_{i-1},\qquad
 H(A_i\mid J,A_{<i})_{\Omega_i}
 \geq\frac1M\sum_jh_j\geq\frac{2-\negl(\lambda)}M.
\end{equation*}
For each window position, the shift makes the source coordinate
uniform. Truncation of the first copy only removes conditioning
systems, which increases conditional entropy by strong subadditivity.
Thus the classical random-start idea continues to work with quantum
side information, including negative conditional entropies.

\paragraph{Asymptotic equipartition of entropy.}
Repeat the row independently $t$ times and group the $i$th window
qubits into a column $C_i$. Its original past is
$P_i=(J^t,C_1,\ldots,C_{i-1})$. For an explicit polynomial $t$, the
conditional entropy bound yields
\begin{equation*}
 H_{\min}^{2^{-\lambda}/12}(C_i\mid P_i)_{\Omega_i^{\otimes t}}
 \geq\frac{7t}{4M}.
\end{equation*}
We use the same predetermined extraction length
$\ell=t/2+3t/(4M)$ in every column. Extraction makes its retained
register $R_i$ close to maximally mixed independently of its original
past, including the public seed.

The security hybrids replace the retained columns one by one, from
last to first. First substitute the witness for the actual prefix;
this preserves computational indistinguishability. Then extract the
current column. Its retained output becomes statistically close to an
independent maximally mixed state, and its remaining past is exactly
the actual preceding prefix. Applying the earlier extraction maps
therefore gives the next hybrid. This order uses each witness only
on the prefix for which it is defined.

\paragraph{Net expansion for 1PRS.}
The actual construction keeps the unused column qubits and those outside
the windows and masks them with independent secret Paulis. Every fixed-key
output remains pure. The quantum one-time-pad identity makes the averaged
output equal to the retained ensemble tensored with an independent
maximally mixed masked register.

Let $q$ count the public seed and shift bits. There are $MLt$ source
qubits, $Lt$ source-key bits, and
$d=N(t-\ell)+Mt$ masked qubits. The output and key lengths satisfy
\begin{equation*}
 n=q+MLt,\qquad r=q+Lt+2d,\qquad
 n-r=\left(M-\frac32\right)t>0.
\end{equation*}
Thus the entropy surplus pays for the source keys and all Pauli masks.
The new ingredients are
retaining the purification to fix the source entropy and applying the
random-start argument to local witnesses with quantum side information.

\paragraph{Reverse implication.}
The implication from $\PRS$ to EFI is known~\cite{my,ccc} and follows
from a simple support-size argument. Given a $\PRS$ with $r$-bit keys
and $n$-qubit outputs $\ket{\phi_k}$, where $n>r$, consider
\begin{equation*}
 \rho_0=2^{-r}\sum_{k\in\{0,1\}^r}\ket{\phi_k}\bra{\phi_k},
 \qquad \rho_1=I/2^n.
\end{equation*}
Both states are efficiently preparable, and single-copy pseudorandomness
makes them computationally indistinguishable. The support of $\rho_0$
has dimension at most $2^r$, whereas $\rho_1$ assigns at most
$2^{r-n}$ weight to this support. Their trace distance is therefore at
least $1-2^{r-n}\geq1/2$, giving an EFI pair.

\section{Preliminaries}\label{sec:prelim}

All logarithms are base two, $\lambda\geq1$ is the security parameter,
and $[m]=\{1,\ldots,m\}$. Register widths and length functions are
polynomially bounded and computable in polynomial time from $1^\lambda$.
Write $\MM{A}=I_A/\dim A$.

\subsection{Algorithms and cryptographic primitives}

\begin{definition}[QPT algorithm]
A (uniform) quantum polynomial-time (QPT) algorithm is a uniform family
$\{C_\lambda\}_{\lambda\in\mathbb N}$ of quantum circuits of size
$\poly(\lambda)$. Uniformity means that a deterministic classical
polynomial-time Turing machine, on input $1^\lambda$, outputs a
description of $C_\lambda$.
\end{definition}

\begin{definition}[Non-uniform QPT algorithm with quantum advice]
A {non-uniform QPT algorithm} (or {non-uniform QPT adversary}) is a family of tuples
$\mathcal{A} = \left\{(C_\lambda, \rho_\lambda)\right\}_{\lambda \in \mathbb{N}}$ such that:
\begin{itemize}
\item $\{C_\lambda\}_{\lambda \in \mathbb{N}}$ is a family of polynomial-size quantum
circuits.
\item $\{\rho_\lambda\}_{\lambda \in \mathbb{N}}$ is a family of {advice states}, where
$\rho_\lambda$ is a quantum state on $\poly(\lambda)$ qubits.
\end{itemize}
 Note that
$\{\rho_\lambda\}_\lambda$ need not be efficiently generatable; $C_\lambda$ receives
$\rho_\lambda$ in a designated  (advice) register. Without loss of generality, the advice state can be assumed to be pure by providing a purification of the advice state.
\end{definition} 

Unless stated otherwise, adversaries are non-uniform QPT algorithms with
quantum advice, the commonly considered adversary model~\cite{kt,bj}.
The advice depends on the security parameter and is independent of the
freshly sampled challenge and its branch.

\begin{definition}[Computational indistinguishability]
Families $\{\rho_{0,\lambda}\}_\lambda$ and
$\{\rho_{1,\lambda}\}_\lambda$ of polynomial-size quantum states are
computationally indistinguishable, written $\rho_0\capprox\rho_1$
or $\rho_0\approx_{\negl_C}\rho_1$, if for
every adversary $\mathcal A$ there is a negligible function $\nu$ such that
\begin{equation*}
 \left|\Pr[\mathcal A(1^\lambda,\rho_{0,\lambda})=1]
       -\Pr[\mathcal A(1^\lambda,\rho_{1,\lambda})=1]\right|
 \leq\nu(\lambda).
\end{equation*}
We denote this absolute difference by
$\adv_{\mathcal A}(\rho_0,\rho_1)$.
\end{definition}

\begin{fact}\label{fact:datanegl}
Let $\{\rho_{0,\lambda}\}_\lambda$ and
$\{\rho_{1,\lambda}\}_\lambda$ be families of states such that
$\rho_0\approx_{\negl_C}\rho_1$, and let $B$ be a non-uniform $\QPT$
algorithm. Then $B(\rho_0)\approx_{\negl_C}B(\rho_1)$.
\end{fact}
\begin{proof}
Assume for contradiction that there exists a non-uniform $\QPT$
algorithm $A$ and a positive polynomial $p(\cdot)$ such that for
infinitely many $\lambda$,
\begin{equation*}
 \Big\vert\Pr[A(1^\lambda,B(\rho_{0,\lambda}))=1]
 -\Pr[A(1^\lambda,B(\rho_{1,\lambda}))=1]\Big\vert
 \geq\frac1{p(\lambda)}.
\end{equation*}
We now create a non-uniform adversary $A'$ to distinguish between
$\rho_{0,\lambda}$ and $\rho_{1,\lambda}$ as follows:
\begin{itemize}
\item $A'$ takes as input the security parameter $1^\lambda$ and
the state $\rho_{b,\lambda}$, for $b\in\{0,1\}$.
\item $A'$ returns $A(1^\lambda,B(\rho_{b,\lambda}))$.
\end{itemize}
The advice of $A'$ consists of the tensor product of the advice states
of $A$ and $B$, so $A'$ is a non-uniform $\QPT$ algorithm. We have
\begin{align*}
 &\Big\vert\Pr[A'(1^\lambda,\rho_{0,\lambda})=1]
 -\Pr[A'(1^\lambda,\rho_{1,\lambda})=1]\Big\vert\\
 &\qquad=\Big\vert\Pr[A(1^\lambda,B(\rho_{0,\lambda}))=1]
 -\Pr[A(1^\lambda,B(\rho_{1,\lambda}))=1]\Big\vert\\
 &\qquad\geq\frac1{p(\lambda)}.
\end{align*}
This contradicts $\rho_0\approx_{\negl_C}\rho_1$.
\end{proof}

\begin{fact}[Indistinguishability of multiple samples]
\label{fact:indistinguishability_multiple_samples}
If $\rho_0\capprox\rho_1$, then
$\rho_0^{\otimes t}\capprox\rho_1^{\otimes t}$ for every
polynomially bounded integer function $t=t(\lambda)$.
\end{fact}
\begin{proof}
Suppose that an adversary $\mathcal A$ distinguishes the tensor powers
with advantage at least $1/p(\lambda)$ for infinitely many $\lambda$.
For each such $\lambda$, define
$H_i=\rho_0^{\otimes(t-i)}\otimes\rho_1^{\otimes i}$,
$0\leq i\leq t$. The triangle inequality gives
\begin{equation*}
 \frac1{p(\lambda)}
 \leq\sum_{i=1}^{t}
 \left|\Pr[\mathcal A(1^\lambda,H_{i-1})=1]
       -\Pr[\mathcal A(1^\lambda,H_i)=1]\right|.
\end{equation*}
Hence some $j\in[t]$ has an adjacent advantage at least
$1/(p(\lambda)t)$. Give a new adversary $\mathcal A'$ the index $j$,
the original advice of $\mathcal A$, and independent copies of
$\rho_0^{\otimes t}$ and $\rho_1^{\otimes t}$ as quantum advice.
On receiving $\rho_c$, it runs $\mathcal A$ on
\begin{equation*}
 \rho_0^{\otimes(t-j)}\otimes\rho_c\otimes
 \rho_1^{\otimes(j-1)}.
\end{equation*}
Its advice has polynomial size and is independent of the challenge.
For $c=0,1$, these inputs are respectively $H_{j-1},H_j$, so
$\mathcal A'$ has inverse-polynomial advantage for infinitely many
$\lambda$, contradicting $\rho_0\capprox\rho_1$.
\end{proof}

\begin{definition}[EFI pairs~\cite{bcq}]\label{def:efi}
An EFI pair is a pair of state families generated by a QPT algorithm
$\mathsf{StateGen}(1^\lambda,b)$, $b\in\{0,1\}$, whose outputs
$\rho_{b,\lambda}$ satisfy:
\begin{enumerate}
\item $\rho_0\capprox\rho_1$;
\item $D(\rho_{0,\lambda},\rho_{1,\lambda})\geq1/p(\lambda)$ for
some polynomial $p$ and all sufficiently large $\lambda$.
\end{enumerate}
\end{definition}

\begin{fact}[EFI amplification~\cite{bqsy}]\label{fact:amplification}
The EFI definition is unchanged if its distance requirement is replaced
by $D(\rho_0,\rho_1)\geq1-\negl(\lambda)$.
\end{fact}

Henceforth an EFI pair has distance $1-\negl(\lambda)$.

\begin{definition}[Single-copy pseudorandom states~\cite{my}]
\label{def:prs}
A single-copy pseudorandom state generator is a QPT algorithm that maps
each $r(\lambda)$-bit key $k$ to a pure $n(\lambda)$-qubit state
$\ket{\phi_{\lambda,k}}$ such that
\begin{equation*}
 \sigma_\lambda
 =2^{-r}\sum_{k\in\{0,1\}^r}\ket{\phi_{\lambda,k}}\bra{\phi_{\lambda,k}}
 \capprox I/2^n.
\end{equation*}
It is expanding if $n(\lambda)>r(\lambda)$. Throughout, $\PRS$ denotes
this expanding variant. A non-uniform generator
($\mathsf{nu\text{-}1PRS}$) may additionally receive a classical advice
string depending on $\lambda$.
\end{definition}

This definition is equivalent to indistinguishability from one copy of
a Haar-random pure state, whose average is $I/2^n$.

\subsection{Quantum information facts}

\begin{definition}[Trace distance]\label{def:trace-distance}
For an operator $A$, its trace norm is
$\|A\|_1:=\Tr\sqrt{A^\dagger A}$.
The trace distance between states $\rho$ and $\sigma$ on the same
Hilbert space is
\begin{equation*}
 D(\rho,\sigma):=\frac12\|\rho-\sigma\|_1.
\end{equation*}
\end{definition}

\begin{fact}[Data processing]\label{fact:data-processing}
For states $\rho,\sigma$ and a completely positive trace-preserving
(CPTP) map $\Phi$,
\begin{equation*}
 D(\Phi(\rho),\Phi(\sigma))\leq D(\rho,\sigma).
\end{equation*}
\end{fact}

\begin{fact}[Trace distance]\label{fact:trace-distance}
Trace distance bounds the difference in acceptance probabilities of any
quantum measurement: for every $0\leq M\leq I$,
\begin{equation*}
 \bigl|\Tr(M\rho)-\Tr(M\sigma)\bigr|\leq D(\rho,\sigma).
\end{equation*}
For states with the same classical distribution,
\begin{equation*}
 D\left(\sum_xp_x\ket x\bra x\otimes\rho_x,
         \sum_xp_x\ket x\bra x\otimes\sigma_x\right)
 =\sum_xp_xD(\rho_x,\sigma_x).
\end{equation*}
\end{fact}

\begin{fact}[Conditional entropy]\label{fact:entropy}
Write $S(A)_\rho=-\Tr\rho_A\log\rho_A$ and
$H(A\mid B)_\rho=S(AB)_\rho-S(B)_\rho$. Conditional entropy obeys
the chain rule, is additive on independent states, and satisfies
$H(A\mid B)\geq H(A\mid BC)$ by strong subadditivity. For classical $X$,
\begin{equation*}
 H(A\mid BX)_{\sum_xp_x\ket x\bra x\otimes\rho^x}
 =\sum_xp_xH(A\mid B)_{\rho^x}.
\end{equation*}
\end{fact}

\begin{fact}[Purifying an efficient preparation]\label{fact:purification}
A QPT preparation of $\rho_R$ can be made into a coherent preparation
of a purification $\ket\psi_{RE}$ by retaining its randomness,
measurement records, and workspace.
\end{fact}

\begin{fact}[Quantum one-time pad]\label{fact:qotp}
For a $d$-qubit register $D$, put
$P_{u,v}=\bigotimes_{j=1}^dX^{u_j}Z^{v_j}$.
Every state $Y_{AD}$ satisfies
\begin{equation*}
 \frac1{4^d}\sum_{u,v\in\{0,1\}^d}
 (I_A\otimes P_{u,v})Y(I_A\otimes P_{u,v}^\dagger)
 =\Tr_D(Y)\otimes\MM{D}.
\end{equation*}
\end{fact}
\begin{proof}
Expand the $D$ component in the Pauli basis. Averaging conjugation
annihilates every nonidentity Pauli and preserves the identity component.
\end{proof}

\subsection{Entropy and extraction tools}

We use the smooth conditional min-entropy convention
of~\cite{tomamichel}, with purified-distance smoothing over
subnormalized states. Namely,
\begin{equation*}
 H_{\min}(A\mid B)_\zeta
 =\sup_{\sigma_B\geq0:\Tr\sigma_B=1}
   \sup\{k:\zeta_{AB}\leq2^{-k}I_A\otimes\sigma_B\},
\end{equation*}
and $H_{\min}^{\delta}$ optimizes over the purified-distance
$\delta$-ball.

\begin{fact}[Conditional AEP~\cite{tomamichel}]\label{fact:aep}
Let $A$ be one qubit and set
$g(\delta)=-\log(1-\sqrt{1-\delta^2})$.
If $m\geq(8/5)g(\delta)$, then
\begin{equation*}
 H_{\min}^{\delta}(A^m\mid B^m)_{\zeta^{\otimes m}}
 \geq mH(A\mid B)_\zeta-10\sqrt{m\,g(\delta)}.
\end{equation*}
The constant is independent of the dimension of $B$.
\end{fact}

\begin{fact}[Conditional Clifford extraction~\cite{decoupling}]
\label{fact:extraction}
Let $A$ have $m$ qubits and write $A=A_1A_2$, with $\ell$ qubits
in $A_1$. If
\begin{equation*}
 H_{\min}^{\varepsilon/12}(A\mid B)_\zeta\geq k,
 \qquad \ell\leq(m+k)/2-\log(1/\varepsilon),
\end{equation*}
then a uniformly random Clifford $U$ satisfies
\begin{equation*}
 \mathbb E_U D\left(\Tr_{A_2}(U\zeta U^\dagger),
                    \MM{A_1}\otimes\zeta_B\right)\leq\varepsilon.
\end{equation*}
The same bound holds when the seed is included as a public classical
register.
\end{fact}

The Clifford group is an exact unitary two-design~\cite{design}, so the
claim follows from one-shot decoupling. The public-seed assertion follows
from Fact~\ref{fact:trace-distance}.

\begin{lemma}[Finite-precision Clifford sampling]\label{lem:clifford-sampling}
Let
\begin{equation*}
 K_t=|\mathcal C_t|=2^{t^2+2t}\prod_{j=1}^t(4^j-1),
 \qquad a=\lceil\log K_t\rceil+\lambda^2,
\end{equation*}
where Clifford operators are counted modulo global phase.
There is a polynomial-time computable map from $a$ bits to an exact
Clifford circuit whose distribution is within $2^{-\lambda^2}$ of
uniform in total variation distance. If uniform Clifford extraction
has error at most $\varepsilon$, publishing these raw seed bits gives
error at most $\varepsilon+2^{-\lambda^2}$ relative to an independent
uniform seed and the ideal extracted state.
\end{lemma}
\begin{proof}
Interpret the seed as an integer, reduce it modulo $K_t$, and use
the efficient Clifford indexing and synthesis of~\cite{clifford}.
The induced distribution has total variation distance at most
$K_t/2^a\leq2^{-\lambda^2}$. Trace distance with the raw seed included
is the average of the conditional trace distances, each at most one
(Fact~\ref{fact:trace-distance}). Changing the Clifford distribution
therefore adds at most $2^{-\lambda^2}$. The raw seed is exactly uniform.
\end{proof}

\section{A source with local prefix witnesses}\label{sec:prefix}

Let $(\rho_0,\rho_1)$ be an EFI pair.
The correlated-flag construction of~\cite{ccc} gives
\begin{equation*}
 \alpha_{RF}=\frac12\sum_{b=0}^1\rho_b\otimes\ket b\bra b_F,
 \qquad
 \beta_{RF}=\frac{\rho_0+\rho_1}{2}\otimes\frac{I_F}{2},
\end{equation*}
with
\begin{equation}
\label{eq:entropic-flag-gap}
 \alpha\capprox\beta,
 \qquad S(\beta)-S(\alpha)=:\chi\geq1-\eta(\lambda),
 \qquad\eta=\negl(\lambda).
\end{equation}
The idea is that a Helstrom measurement guesses the flag with error
$e=(1-D(\rho_0,\rho_1))/2$; data processing and Fano's inequality give
$\chi=I(F:R)_\alpha\geq1-h_2(e)$, where $h_2$ is binary entropy.
Thus we may take $\eta=h_2(e)=\negl(\lambda)$.
By Fact~\ref{fact:purification}, prepare purifications
$\ket{\psi_b}_{RE}$ of $\rho_b$. Define the pure
one-bit source
\begin{equation*}
 G(b)=\ket b_F\ket{\psi_b}_{RE}\ket0_{\mathrm{pad}}.
\end{equation*}
Order its qubits as $R,F,E$, followed by padding, and pad to a
power-of-two width $M\geq4$. Write the ordered qubits as
$Q_1,\ldots,Q_M$, with $Q_f=F$. Its average
\begin{equation*}
 \tau=\tfrac12\sum_b G(b)G(b)^\dagger
\end{equation*}
has entropy exactly one: its two pure summands have orthogonal flags.
Let $\tau^{[j]}$ be its first-$j$-qubit marginal and set
$\tau^{[0]}=1$. In particular, $\tau^{[f]}=\alpha$.

\begin{lemma}[Local witnesses and their entropy sum]
\label{lem:prefix-witnesses}\label{lem:entropy-sum}
Define
\begin{equation*}
 \gamma_j=\tau^{[j]}\quad(j\ne f),\qquad\gamma_f=\beta.
\end{equation*}
These states are efficiently preparable and satisfy
\begin{equation}
\label{eq:prefix-same}
 \gamma_j\capprox\tau^{[j]},\qquad
 \Tr_{Q_j}\gamma_j=\tau^{[j-1]}.
\end{equation}
For $h_j=H(Q_j\mid Q_{<j})_{\gamma_j}$,
\begin{equation*}
 \sum_{j=1}^M h_j=1+\chi\geq2-\eta(\lambda).
\end{equation*}
\end{lemma}
\begin{proof}
All assertions except at $j=f$ follow from equality with the actual
prefix. At the flag, hiding follows from~\eqref{eq:entropic-flag-gap},
and the preceding marginal is unchanged because
$\alpha_R=\beta_R$. Thus replacing the flag term in the actual
entropy chain (Fact~\ref{fact:entropy}) changes it by $S(\beta)-S(\alpha)=\chi$. Hence
\begin{equation*}
 \sum_j h_j=S(\tau)+S(\beta)-S(\alpha)=1+\chi\geq2-\eta.
\end{equation*}
\end{proof}

The $\gamma_j$ are separate witnesses for separate prefixes, not
marginals of one common state. They agree with the \emph{actual}
previous-prefix marginal; individual $h_j$ may be negative.

\section{Random-start idea}\label{sec:equalization}

Use the fixed parameters
\begin{equation*}
 L=4M,\quad N=M(L-1),\quad
 t=2^{16}M^2(\lambda+4),\quad
 \ell=\frac t2+\frac{3t}{4M}.
\end{equation*}
They are polynomial-time computable integers, and $\ell\leq t$.
In one row, concatenate $L$ independently keyed copies of $G$.
Choose a public uniform $J\in\{0,\ldots,M-1\}$ using $\log M$ bits.
The main window is
\begin{equation*}
 J+1,\ldots,J+N,
\end{equation*}
with qubits relabeled $A_1,\ldots,A_N$. Its $M$ complementary qubits
are omitted from the prefix analysis and retained for masking by the
construction. Let $\omega_i$ denote the actual state on
$J,A_1,\ldots,A_i$, with $\omega_0=I_J/M$.

For fixed $i,J$, the current source copy and coordinate are
\begin{equation*}
 k=\left\lfloor\frac{J+i-1}{M}\right\rfloor+1,
 \qquad j=((J+i-1)\bmod M)+1.
\end{equation*}
Let $\Lambda_{i,J}$ be the actual preceding-copy factor: the first
copy's tail marginal and any complete intervening copies. It is the
trivial state when $k=1$. Define
\begin{equation*}
 \Theta_{i,J}=\begin{cases}
  \Tr_{Q_1\cdots Q_J}\gamma_j,&k=1,\\
  \gamma_j,&k\geq2,
 \end{cases}
\end{equation*}
and, after identifying source qubits with their window registers,
\begin{equation*}
 \Omega_i=\frac1M\sum_J\ket J\bra J\otimes
                              \Lambda_{i,J}\otimes\Theta_{i,J}.
\end{equation*}
Thus only the current copy is replaced; all earlier factors are
actual source marginals.

\begin{lemma}\label{lem:equalization}
For every $i\in[N]$, the efficiently preparable witness $\Omega_i$
satisfies
\begin{equation*}
 \Tr_{A_i}\Omega_i=\omega_{i-1},\qquad\Omega_i\capprox\omega_i,
\end{equation*}
and
\begin{equation*}
 H(A_i\mid J,A_{<i})_{\Omega_i}
 \geq\frac1M\sum_jh_j\geq\frac{2-\eta}{M}.
\end{equation*}
\end{lemma}
\begin{proof}
The exact marginal identity follows from~\eqref{eq:prefix-same}.
Only the shift $J=(f-i)\bmod M$ changes the actual prefix. There the
current copy contains the available part of $R$ and its flag, with
its purification $E$ absent, so~\eqref{eq:entropic-flag-gap} applies
after partial trace and adjoining independent preceding factors
(Fact~\ref{fact:datanegl}).
Conditioned on $J$, those factors do not affect conditional entropy.
The current-copy entropy is $h_j$ when its full past is present;
truncating the first copy removes conditioning systems and can only
increase it by strong subadditivity (Fact~\ref{fact:entropy}). Finally, $J\mapsto j$ is a
permutation of $[M]$, so averaging over the public $J$ gives the bound.
\end{proof}

Prepare $t$ independent rows, including independent public shifts.
Group their qubits in position $i$ into $C_i$, and put
$P_i=(J^t,C_1,\ldots,C_{i-1})$. Up to a register permutation, the
actual prefix and its witness are $\omega_i^{\otimes t}$ and
$\Omega_i^{\otimes t}$. Their indistinguishability follows from
Fact~\ref{fact:indistinguishability_multiple_samples}.

\begin{lemma}\label{lem:column-entropy}
For $\epsilon=2^{-\lambda}/12$ and all sufficiently large $\lambda$,
every column satisfies
\begin{equation*}
 H_{\min}^{\epsilon}(C_i\mid P_i)_{\Omega_i^{\otimes t}}
 \geq k:=\frac{7t}{4M},
\end{equation*}
and
\begin{equation*}
 \frac{t+k}{2}-\ell=\frac{t}{8M}\geq\lambda.
\end{equation*}
\end{lemma}
\begin{proof}
Apply Fact~\ref{fact:aep} to
$\Omega_i$, with side information $J,A_{<i}$. For
$g(\epsilon)=-\log(1-\sqrt{1-\epsilon^2})\leq3(\lambda+4)$,
our $t$ satisfies its hypothesis $t\geq(8/5)g(\epsilon)$ and gives
\begin{equation*}
 H_{\min}^{\epsilon}(C_i\mid P_i)
 \geq\frac{t(2-\eta)}M-10\sqrt{t\,g(\epsilon)}
 \geq\frac{15t}{8M}-\frac{t}{8M}=\frac{7t}{4M},
\end{equation*}
where $\eta\leq1/8$ eventually and
$10\sqrt{3t(\lambda+4)}\leq t/(8M)$. Substitution gives the margin.
\end{proof}

\section{Construction and security}

We first prove that the retained outputs are computationally
indistinguishable from maximally mixed. We then introduce the quantum one-time pad: we keep every
qubit traced out in the analysis and mask it with independent secret
Paulis. We then show that overall stretch is still positive, which gives us a $\PRS$.

\begin{theorem}\label{thm:uniform-generation}
EFI pairs exist if and only if expanding single-copy pseudorandom state
generators exist.
\end{theorem}

We prove the forward implication below. The reverse implication is
known~\cite{my,ccc}: an expanding pure-state generator has an efficiently
generated average output of rank at most $2^r$, at trace distance at least
$1-2^{r-n}\geq1/2$ from the maximally mixed state. These states satisfy
Definition~\ref{def:efi}; Fact~\ref{fact:amplification} gives the amplified form.

\subsection{Security of the retained outputs}

Use the parameters and prefix states from Section~\ref{sec:equalization}. Applying the
seeded Clifford $U_i(s_i)$ to column $C_i$ produces registers $R_iD_i$,
where $R_i$ contains $\ell$ qubits. For the analysis, define
\begin{equation*}
 \Phi_{i,s_i}(X)
 =
 \Tr_{D_i}\!\left(U_i(s_i)XU_i(s_i)^\dagger\right).
\end{equation*}
The actual construction keeps $D_i$; the partial trace is used only
in the security argument.

Let $S=(s_1,\ldots,s_N)$ contain the independent raw binary Clifford
seeds, each of length $a$. Write
\begin{equation*}
 X_i=\omega_i^{\otimes t},
 \qquad
 Z_i=\Omega_i^{\otimes t},
\end{equation*}
with the row-to-column register permutation understood. Recall that
\begin{equation*}
 P_i=(J^t,C_1,\ldots,C_{i-1}),
 \qquad
 (Z_i)_{P_i}=X_{i-1}.
\end{equation*}
The witness therefore has exactly the actual preceding marginal.

For $i=0,\ldots,N$, define
\begin{equation}
\label{eq:hybrids}
 \mathcal H_i
 =
 2^{-Na}\sum_S \ket S\bra S\otimes
 \bigl(\Phi_{1,s_1}\otimes\cdots\otimes\Phi_{i,s_i}\bigr)(X_i)
 \otimes \MM{R_{i+1}\cdots R_N}.
\end{equation}
The shifts $J^t$ are included in $X_i$. All seeds are included even
when their column outputs have been replaced by maximally mixed
registers. Thus
\begin{equation*}
 \mathcal H_N=\text{actual retained-output ensemble},
 \qquad
 \mathcal H_0=\MM{S,J^t,R_1,\ldots,R_N}.
\end{equation*}
Let $\mathcal W_i$ be defined as in~\eqref{eq:hybrids}, with $X_i$
replaced by $Z_i$. Each transition has the form
\begin{equation*}
 \mathcal H_i\capprox\mathcal W_i\sapprox\mathcal H_{i-1},
\end{equation*}
where $\sapprox$ denotes negligible trace distance.
First, by Lemma~\ref{lem:equalization} and
Facts~\ref{fact:indistinguishability_multiple_samples} and~\ref{fact:datanegl},
\begin{equation*}
 \mathcal H_i\capprox\mathcal W_i.
\end{equation*}
Second, extract $C_i$ while retaining
its original past $P_i$. Lemma~\ref{lem:column-entropy},
Fact~\ref{fact:extraction}, and Lemma~\ref{lem:clifford-sampling} give
the following bound for the seeded extraction channel $\operatorname{Ext}_i$:
\begin{equation*}
 D\left(\operatorname{Ext}_i(Z_i),
         \MM{s_iR_i}\otimes(Z_i)_{P_i}\right)
 \leq e_\lambda:=2^{-\lambda}+2^{-\lambda^2}.
\end{equation*}
This includes the public raw seed $s_i$. Extraction makes $R_i$ uniform
and independent of the past, while the identity
\begin{equation*}
 (Z_i)_{P_i}=X_{i-1}
\end{equation*}
identifies that past with the actual preceding prefix. Applying the
earlier extraction maps and appending the uniform future outputs,
trace-distance contraction (Fact~\ref{fact:data-processing}) gives
\begin{equation}
\label{eq:hybrid-statistical}
 D(\mathcal W_i,\mathcal H_{i-1})\leq e_\lambda.
\end{equation}
For any adversary $\mathcal T$, the triangle inequality now gives
\begin{equation*}
 \adv_{\mathcal T}(\mathcal H_N,\mathcal H_0)
 \leq\sum_{i=1}^{N}\adv_{\mathcal T}(\mathcal H_i,\mathcal W_i)
       +Ne_\lambda.
\end{equation*}
The computational contribution is negligible by a standard non-uniform hybrid argument. Otherwise, one could select a column and a row with non-negligible distinguishing advantage, contradicting EFI hiding. Since the statistical contribution $Ne_\lambda$ is also negligible, we obtain $\mathcal H_N\capprox\mathcal H_0$.

\subsection{Pure outputs using a quantum one-time pad}

We now turn the retained-output analysis into a pure-state
generator. Define
\begin{equation*}
 A=(S,J^t,R_1,\ldots,R_N),
 \qquad
 D=D_1\cdots D_N B,
\end{equation*}
where $B$ contains every boundary qubit outside the
selected windows. The number of qubits in $D$ is
\begin{equation*}
 d=N(t-\ell)+Mt=MLt-N\ell.
\end{equation*}
Let $x=(S,J^t,b)$ collect the seeds, shifts, and $Lt$ source bits.
For fixed $x$, prepare all pure source copies, arrange the selected
windows into columns, apply the column Cliffords, and include
the public basis registers $\ket{S,J^t}$. Keep all qubits.
Denote the resulting pure state by
\begin{equation*}
 \ket{\Psi_x}_{AD}.
\end{equation*}
Its averaged retained marginal is exactly
\begin{equation*}
 \mathbb E_x\Tr_D\ket{\Psi_x}\bra{\Psi_x}
 =
 \mathcal H_N.
\end{equation*}
Choose an independent secret mask key
$z=(u,v)\in\{0,1\}^{2d}$ and define
\begin{equation*}
 P_z=\bigotimes_{j=1}^d X^{u_j}Z^{v_j}.
\end{equation*}
The complete key is $k=(x,z)$, and the generator outputs
\begin{equation*}
 \boxed{
 \ket{\mathsf{Gen}_\lambda(x,z)}
 =
 (I_A\otimes P_z)\ket{\Psi_x}.
 }
\end{equation*}
Every fixed-key output is pure because all preparation and masking
operations are unitary.

By Fact~\ref{fact:qotp}, averaging the mask replaces $D$ by an
independent maximally mixed state. Averaging first over the mask and then over $x$ gives
\begin{equation*}
 \mathbb E_{x,z}
 \ket{\mathsf{Gen}_\lambda(x,z)}
 \bra{\mathsf{Gen}_\lambda(x,z)}
 =
 \mathcal H_N\otimes\MM{D}.
\end{equation*}
Therefore,
\begin{equation*}
 \mathcal H_N\otimes\MM{D}
 \capprox
 \mathcal H_0\otimes\MM{D}
 =
 \MM{AD}.
\end{equation*}
Indeed, a full-output test can be simulated by a retained-output
test that prepares an independent $\MM{D}$.

\subsection{Key length and positive stretch}

The public seeds and shifts use
\begin{equation*}
 q=Na+t\log M
\end{equation*}
bits, contributing equally to key and output length.
There are $Lt$ source bits and $2d$ secret OTP bits. Hence
\begin{equation*}
 r=q+Lt+2d,
 \qquad
 n=q+MLt.
\end{equation*}
Using $d=N(t-\ell)+Mt$,
\begin{equation*}
 n-r=N(2\ell-t)-Mt-Lt.
\end{equation*}
Substituting the fixed parameters
\begin{equation*}
 N=M(L-1),\qquad
 L=4M,\qquad
 2\ell-t=\frac{3t}{2M},
\end{equation*}
we obtain
\begin{equation*}
 \boxed{
 n-r
 =
 \left(\frac L2-M-\frac32\right)t
 =
 \left(M-\frac32\right)t
 >0.
 }
\end{equation*}

\section*{Acknowledgements}
RB would like to thank Rahul Jain and Bo-Cheng Wu for helpful discussions at CQT, NUS.

\section*{AI disclosure}
The author proposed the purification and random-start idea to ChatGPT
Astra, which helped develop a rigorous proof of the lemma and assisted in
drafting the paper. The author verified and reorganized these arguments
to obtain the final results and retains full responsibility for the
contents of the paper.


\begin{thebibliography}{CCC$^+$26}

\bibitem[AQY22]{aqy}
Prabhanjan Ananth, Luowen Qian, and Henry Yuen.
\newblock Cryptography from pseudorandom quantum states.
\newblock In Yevgeniy Dodis and Thomas Shrimpton, editors,
\emph{Advances in Cryptology---CRYPTO 2022, Part I}, volume 13507
of Lecture Notes in Computer Science, pages 208--236.
Springer, Cham, 2022.
\newblock \url{https://doi.org/10.1007/978-3-031-15802-5_8}.

\bibitem[BJ24]{bj}
Rishabh Batra and Rahul Jain.
\newblock Commitments are equivalent to statistically-verifiable one-way
state generators.
\newblock In \emph{2024 IEEE 65th Annual Symposium on Foundations of
Computer Science (FOCS)}, pages 1178--1192.
IEEE Computer Society, 2024.
\newblock \url{https://doi.org/10.1109/FOCS61266.2024.00077}.
\newblock Full version: \url{https://arxiv.org/abs/2404.03220}.

\bibitem[BQSY24]{bqsy}
John Bostanci, Luowen Qian, Nicholas Spooner, and Henry Yuen.
\newblock An efficient quantum parallel repetition theorem and applications.
\newblock In \emph{Proceedings of the 56th Annual ACM Symposium on
Theory of Computing (STOC)}, pages 1478--1487. ACM, New York, 2024.
\newblock \url{https://doi.org/10.1145/3618260.3649603}.

\bibitem[BCQ23]{bcq}
Zvika Brakerski, Ran Canetti, and Luowen Qian.
\newblock On the computational hardness needed for quantum cryptography.
\newblock In Yael Tauman Kalai, editor,
\emph{14th Innovations in Theoretical Computer Science Conference
(ITCS 2023)}, volume 251 of Leibniz International Proceedings in
Informatics, pages 24:1--24:21.
Schloss Dagstuhl---Leibniz-Zentrum f{\"u}r Informatik, Dagstuhl, 2023.
\newblock \url{https://doi.org/10.4230/LIPIcs.ITCS.2023.24}.

\bibitem[CCC$^+$26]{ccc}
Bruno Cavalar, Boyang Chen, Andrea Coladangelo, Matthew Gray, Zihan Hu,
Zhengfeng Ji, and Xingjian Li.
\newblock A meta-complexity characterization of minimal quantum cryptography.
\newblock In \emph{Proceedings of the 58th Annual ACM Symposium on
Theory of Computing (STOC)}, pages 675--686. ACM, New York, 2026.
\newblock \url{https://doi.org/10.1145/3798129.3800783}.
\newblock Full version: \url{https://arxiv.org/abs/2510.07859} (2025).

\bibitem[DCEL09]{design}
Christoph Dankert, Richard Cleve, Joseph Emerson, and Etera Livine.
\newblock Exact and approximate unitary 2-designs and their application
to fidelity estimation.
\newblock \emph{Physical Review A}, 80(1):012304, 2009.
\newblock \url{https://doi.org/10.1103/PhysRevA.80.012304}.

\bibitem[DBWR14]{decoupling}
Fr{\'e}d{\'e}ric Dupuis, Mario Berta, J{\"u}rg Wullschleger,
and Renato Renner.
\newblock One-shot decoupling.
\newblock \emph{Communications in Mathematical Physics},
328(1):251--284, 2014.
\newblock \url{https://doi.org/10.1007/s00220-014-1990-4}.

\bibitem[G{\'a}c01]{gacs}
P{\'e}ter G{\'a}cs.
\newblock Quantum algorithmic entropy.
\newblock \emph{Journal of Physics A: Mathematical and General},
34(35):6859--6880, 2001.
\newblock \url{https://doi.org/10.1088/0305-4470/34/35/312}.

\bibitem[GMMY24]{gmmy}
Eli Goldin, Tomoyuki Morimae, Saachi Mutreja, and Takashi Yamakawa.
\newblock CountCrypt: Quantum cryptography between QCMA and PP.
\newblock arXiv:2410.14792, 2024; revised October 2025.
\newblock Cryptology ePrint Archive, Paper 2024/1707.
\newblock \url{https://doi.org/10.48550/arXiv.2410.14792};
\url{https://eprint.iacr.org/2024/1707}.

\bibitem[Gol90]{gol90}
Oded Goldreich.
\newblock A note on computational indistinguishability.
\newblock \emph{Information Processing Letters}, 34(6):277--281, 1990.
\newblock \url{https://doi.org/10.1016/0020-0190(90)90010-U}.

\bibitem[Gol01]{gol01}
Oded Goldreich.
\newblock \emph{Foundations of Cryptography, Volume 1: Basic Tools}.
\newblock Cambridge University Press, Cambridge, 2001.
\newblock \url{https://doi.org/10.1017/CBO9780511546891}.

\bibitem[HRV10]{hrv}
Iftach Haitner, Omer Reingold, and Salil Vadhan.
\newblock Efficiency improvements in constructing pseudorandom generators
from one-way functions.
\newblock In \emph{Proceedings of the 42nd ACM Symposium on Theory of
Computing (STOC)}, pages 437--446. ACM, New York, 2010.
\newblock \url{https://doi.org/10.1145/1806689.1806750}.
\newblock Journal version: \emph{SIAM Journal on Computing},
42(3):1405--1430, 2013,
\url{https://doi.org/10.1137/100814421}.

\bibitem[HILL99]{hill}
Johan H{\aa}stad, Russell Impagliazzo, Leonid A. Levin, and Michael Luby.
\newblock A pseudorandom generator from any one-way function.
\newblock \emph{SIAM Journal on Computing}, 28(4):1364--1396, 1999.
\newblock \url{https://doi.org/10.1137/S0097539793244708}.

\bibitem[HN23]{hn}
Shuichi Hirahara and Mikito Nanashima.
\newblock Learning in Pessiland via inductive inference.
\newblock In \emph{2023 IEEE 64th Annual Symposium on Foundations of
Computer Science (FOCS)}, pages 447--457.
IEEE Computer Society, 2023.
\newblock \url{https://doi.org/10.1109/FOCS57990.2023.00033}.

\bibitem[IRS22]{irs}
Rahul Ilango, Hanlin Ren, and Rahul Santhanam.
\newblock Robustness of average-case meta-complexity via pseudorandomness.
\newblock In \emph{Proceedings of the 54th Annual ACM SIGACT Symposium
on Theory of Computing (STOC)}, pages 1575--1583.
ACM, New York, 2022.
\newblock \url{https://doi.org/10.1145/3519935.3520051}.

\bibitem[IL89]{il}
Russell Impagliazzo and Michael Luby.
\newblock One-way functions are essential for complexity based cryptography.
\newblock In \emph{30th Annual Symposium on Foundations of Computer
Science (FOCS)}, pages 230--235. IEEE Computer Society, 1989.
\newblock \url{https://doi.org/10.1109/SFCS.1989.63483}.

\bibitem[JLS18]{jls}
Zhengfeng Ji, Yi-Kai Liu, and Fang Song.
\newblock Pseudorandom quantum states.
\newblock In Hovav Shacham and Alexandra Boldyreva, editors,
\emph{Advances in Cryptology---CRYPTO 2018, Part III}, volume 10993
of Lecture Notes in Computer Science, pages 126--152.
Springer, Cham, 2018.
\newblock \url{https://doi.org/10.1007/978-3-319-96878-0_5}.

\bibitem[KT24]{kt}
Dakshita Khurana and Kabir Tomer.
\newblock Commitments from quantum one-wayness.
\newblock In \emph{Proceedings of the 56th Annual ACM Symposium on
Theory of Computing (STOC)}, pages 968--978. ACM, New York, 2024.
\newblock \url{https://doi.org/10.1145/3618260.3649654}.
\newblock Full version: \url{https://arxiv.org/abs/2310.11526}.

\bibitem[KS14]{clifford}
Robert Koenig and John A. Smolin.
\newblock How to efficiently select an arbitrary Clifford group element.
\newblock \emph{Journal of Mathematical Physics}, 55(12):122202, 2014.
\newblock \url{https://doi.org/10.1063/1.4903507}.

\bibitem[Kre21]{kre}
William Kretschmer.
\newblock Quantum pseudorandomness and classical complexity.
\newblock In Min-Hsiu Hsieh, editor,
\emph{16th Conference on the Theory of Quantum Computation,
Communication and Cryptography (TQC 2021)}, volume 197 of
Leibniz International Proceedings in Informatics, pages 2:1--2:20.
Schloss Dagstuhl---Leibniz-Zentrum f{\"u}r Informatik, Dagstuhl, 2021.
\newblock \url{https://doi.org/10.4230/LIPIcs.TQC.2021.2}.

\bibitem[KQST23]{kqst}
William Kretschmer, Luowen Qian, Makrand Sinha, and Avishay Tal.
\newblock Quantum cryptography in Algorithmica.
\newblock In \emph{Proceedings of the 55th Annual ACM Symposium on
Theory of Computing (STOC)}, pages 1589--1602. ACM, New York, 2023.
\newblock \url{https://doi.org/10.1145/3564246.3585225}.

\bibitem[LP20]{lp}
Yanyi Liu and Rafael Pass.
\newblock On one-way functions and Kolmogorov complexity.
\newblock In \emph{2020 IEEE 61st Annual Symposium on Foundations of
Computer Science (FOCS)}, pages 1243--1254.
IEEE Computer Society, 2020.
\newblock \url{https://doi.org/10.1109/FOCS46700.2020.00118}.

\bibitem[MY22]{my}
Tomoyuki Morimae and Takashi Yamakawa.
\newblock Quantum commitments and signatures without one-way functions.
\newblock In Yevgeniy Dodis and Thomas Shrimpton, editors,
\emph{Advances in Cryptology---CRYPTO 2022, Part I}, volume 13507
of Lecture Notes in Computer Science, pages 269--295.
Springer, Cham, 2022.
\newblock \url{https://doi.org/10.1007/978-3-031-15802-5_10}.

\bibitem[Nao91]{naor}
Moni Naor.
\newblock Bit commitment using pseudorandomness.
\newblock \emph{Journal of Cryptology}, 4(2):151--158, 1991.
\newblock \url{https://doi.org/10.1007/BF00196774}.

\bibitem[Sat26]{microcrypt}
Or Sattath, maintainer.
\newblock Microcrypt Zoo.
\newblock Online reference graph of quantum cryptographic primitives.
Accessed October 13, 2026.
\newblock \url{https://sattath.github.io/microcrypt-zoo/}.

\bibitem[Tom12]{tomamichel}
Marco Tomamichel.
\newblock \emph{A Framework for Non-Asymptotic Quantum Information Theory}.
\newblock PhD thesis, ETH Zurich, 2012.
\newblock \url{https://doi.org/10.3929/ethz-a-7356080}.
\newblock Also available as \url{https://arxiv.org/abs/1203.2142}.

\end{thebibliography}
\end{document}